\documentclass[aps,prx,bibnotes,reprint,superscriptaddress,nofootinbib]{revtex4-2}

\usepackage{graphicx}
\usepackage{dcolumn}
\usepackage{bm}
\usepackage[colorlinks,linkcolor=blue,urlcolor=blue, citecolor=blue]{hyperref}

\usepackage{dsfont}
\usepackage{physics}
\usepackage{enumitem}
\usepackage{tcolorbox}
\usepackage{needspace}

\usepackage{thmtools, thm-restate, amsthm} 

\newtheorem*{theorem*}{Theorem}

\usepackage{amsmath, amssymb, amsfonts, amsthm}

\makeatletter 
    
\renewcommand\onecolumngrid{
\do@columngrid{one}{\@ne}%
\def\set@footnotewidth{\onecolumngrid}
\def\footnoterule{\kern-6pt\hrule width 1.5in\kern6pt}%
}

\renewcommand\twocolumngrid{
        \def\footnoterule{
        \dimen@\skip\footins\divide\dimen@\thr@@
        \kern-\dimen@\hrule width.5in\kern\dimen@}
        \do@columngrid{mlt}{\tw@}
}%

\makeatother    

\begin{document}

\title{Fundamental accuracy-resolution trade-off for timekeeping devices}

\author{Florian Meier}
\email[]{florianmeier256@gmail.com}
\affiliation{Vienna Center for Quantum Science and Technology, Atominstitut, Technische Universität Wien, 1020 Vienna, Austria}
\affiliation{Institute for Theoretical Physics, ETH Zurich, 8093 Z\"{u}rich, Switzerland}
\author{Emanuel Schwarzhans}
\email[]{e.schwarzhans@gmx.at}
\affiliation{Vienna Center for Quantum Science and Technology, Atominstitut, Technische Universität Wien, 1020 Vienna, Austria}
\author{Paul Erker}
\email[]{paul.erker@tuwien.ac.at}
\affiliation{Vienna Center for Quantum Science and Technology, Atominstitut, Technische Universität Wien, 1020 Vienna, Austria}
\affiliation{Institute for Quantum Optics and Quantum Information, Austrian Academy of Sciences, 1090 Vienna, Austria}
\author{Marcus Huber}
\email[]{marcus.huber@tuwien.ac.at}
\affiliation{Vienna Center for Quantum Science and Technology, Atominstitut, Technische Universität Wien, 1020 Vienna, Austria}
\affiliation{Institute for Quantum Optics and Quantum Information, Austrian Academy of Sciences, 1090 Vienna, Austria}

\date{\today}

\begin{abstract}
From a thermodynamic point of view, all clocks are driven by irreversible processes. Additionally, one can use oscillatory systems to temporally modulate the thermodynamic flux towards equilibrium. Focusing on the most elementary thermalization events, this modulation can be thought of as a temporal probability concentration for these events.
There are two fundamental factors limiting the performance of clocks: On the one level, the inevitable drifts of the oscillatory system, which are addressed by finding stable atomic or nuclear transitions that lead to astounding precision of today's clocks. On the other level, there is the intrinsically stochastic nature of the irreversible events upon which the clock's operation is based. This becomes relevant when seeking to maximize a clock's resolution at high accuracy, which is ultimately limited by the number of such stochastic events per reference time unit.
We address this essential trade-off between clock accuracy and resolution, proving a universal bound for all clocks whose elementary thermalization events are memoryless.
\end{abstract}

\maketitle
\raggedbottom

Clocks dominate our daily lives unlike any other technology -- from ordinary things like catching the train in the morning to locating ourselves using the GPS -- they are involved everywhere. Physics has provided the theoretical and experimental foundation to develop accurate and stable clocks which has culminated in the 50s with the invention of atomic clocks \cite{Essen1955}. Year by year, these clocks have been improving their accuracy, with a state-of-the-art optical clock accumulating less than a hundred milliseconds of error over the lifespan of the sun \cite{Ludlow2015,Riehle2015,Bothwell2022}.
These advancements beg for the question whether there are any physical principles constraining a clock's performance. 
Driven by this open problem, quantum clocks are emerging into an independent field of research, unifying approaches from quantum information theory \cite{Rankovic2015,Renner2017,Stupar2018,Woods2021} and quantum thermodynamics \cite{Erker2017,Schwarzhans2021,Milburn2020,Manikandan2022,Cilluffo2022}. Recent experiments showcase today's technology is capable of exploring these ultimate limits of timekeeping \cite{Pearson2021,Berholts2022,He2022}.

This letter explores fundamental limitations of clocks coming from thermodynamics, because clocks, like all other physical systems, are subject to thermodynamical laws \cite{Milburn2020}. Even worse, they very much rely on the increase of entropy originating in the second law of thermodynamics.
This implies that (1) clocks are witnesses of the macroscopic breaking of time-reversal symmetry because they tick forwards in time. (2)~Thus, they must be driven by irreversible processes that drain out-of-equilibrium resources to output temporal information. These processes, though, are thermodynamic, therefore inherently stochastic and never perfectly predictable. (3)~We conclude, even an idealized clock could never be perfect simply due to the fact that the clock itself is fundamentally driven by stochastic processes.

We look at clocks that use this flow by counting elementary thermalization events to define \textit{ticks}. Instructive examples of such a process could range from grains of sand that pass through an hourglass to the photons that are reflected from a pendulum to ascertain it's position. Even modern atomic clocks ultimately rely on a macroscopic number of photons to read out the laser frequency (which is stabilized by feedback from atomic transitions). In all these cases, the underlying thermalization processes are stochastic and have some intrinsic rate $\Gamma$. 
The time passing between two events, relative to assumed smooth parameter time, is probabilistic and described by a probability density function. We refer to its average as $\mu$ and its uncertainty as $\sigma$.
We define the \textit{resolution} $\nu$ of the clock to be the inverse average time between two ticks $\nu=1/\mu$, and the \textit{accuracy} $N$ as the average number of times the clock ticks until it is off by one tick. For a sequence of independent and identically distributed ticks, the accuracy equals the signal-to-noise ratio $N=\mu^2/\sigma^2$ \cite{Erker2017}.
A clock usually modulates the temporal distribution of when the thermalization events occur as to improve its accuracy. For most practical clocks, the stability of this temporal modulation (for instance the laser frequency in atomic clocks) is the limiting factor to the accuracy. We show that even if this temporal modulation is perfectly stable, the stochastic process underlying the tick generation limits the accuracy at given resolution. In other words, the need to irreversibly generate a signal bounds the clock's performance even if it uses an eternally stable frequency reference.

\begin{restatable}[Accuracy-resolution trade-off]{theorem*}{nrbound}
Clocks with elementary ticking events generated by a memoryless stochastic process at rate $\Gamma$ obey the trade-off relation
\begin{align}\label{eq:nr_bound}
        N\leq \frac{\Gamma^2}{\nu^2}.
\end{align}
\end{restatable}

This trade-off complements other established results in the field considering restrictions imposed on the clock through entropy production \cite{Barato2016,Erker2017} or Hilbertspace dimension \cite{Woods2022,Yang2020}.

\paragraph*{Temporal probability concentration.}
The most primitive clock consists of two out-of-equilibrium thermal reservoirs in contact with each other. Counting the individual stochastic thermalization events as ticks can serve as a way to measure time, we call this the \textit{thermal reference clock}.
Sufficiently large reservoirs are memoryless, therefore, such stochastic jumps at equal rate are exponentially distributed \cite{Ingarden1997,Breuer2007,Klenke2020}. The coupling of the two baths defines a characteristic thermalization rate $\Gamma$. For an exponential distribution the standard deviation equals the average, which leads to $\mu=\sigma=\Gamma^{-1}.$
Consequently, such a clock has unit accuracy $N=1$ and only by averaging over many of these ticks are we able to achieve higher accuracy, but at the expense of resolution. Averaging over $M$ independent and identically distributed (i.i.d.) such events increases the variance and mean of the tick time $M$-fold, leading to a resolution $\nu\propto 1/M$ and accuracy $N\propto M$. This gives an inverse proportional accuracy-resolution scaling
\begin{align}\label{eq:ccg_bound}
    N= \frac{\Gamma}{\nu},
\end{align}
quadratically smaller than the upper bound in eq.~\eqref{eq:nr_bound}.
Upon closer inspection, we find a thermodynamic cost associated to this increase in accuracy: instead of having a single irreversible event producing a tick, now $M$ irreversible events are required.

A natural question to ask is whether it is possible to increase the accuracy beyond what is achievable through averaging in eq.~\eqref{eq:ccg_bound}, while still using the same underlying stochastic tick generating process. The answer lies in the observation that all ticking clocks known to us use a combination of two processes to tell time:\begin{enumerate}[label=(\alph*),itemsep=0mm]
    \item irreversible processes that generate ticks, and
    \item a filter process, \textit{temporal probability concentration}, which modulates the probability of the irreversible ticking events to occur.
\end{enumerate}
By means of temporal probability concentration (TPC), a (sensible) clock centers the probability distribution of the stochastic events with a periodic process such that the ticks occur closely around well-defined instants in time.
Exponential decay of an unstable two-level system is an example for a clock without TPC where the decay defines the tick.
The tick probability density conditioned on the tick not yet having occurred equals $\Gamma$, the decay constant which is time-independent.
Such a clock has accuracy $N=1.$
Clockworks modulate this probability by, for example, driving the two-level system from the ground state into the excited state.
This gives an effective time-dependency of the ticking probability as illustrated in Fig.~\ref{fig:TPC_good_sampling}, and this driving is what we call TPC and in general, it gives an accuracy $N>1$. 
The key is to perform the driving autonomously without external temporal control.
In the following we introduce a model for this with more details on the formalism of TPC in Sec.~\ref{sec:temporal_probability_concentration} of the Appendix.
\begin{figure}
    \centering
    \includegraphics[width=\columnwidth]{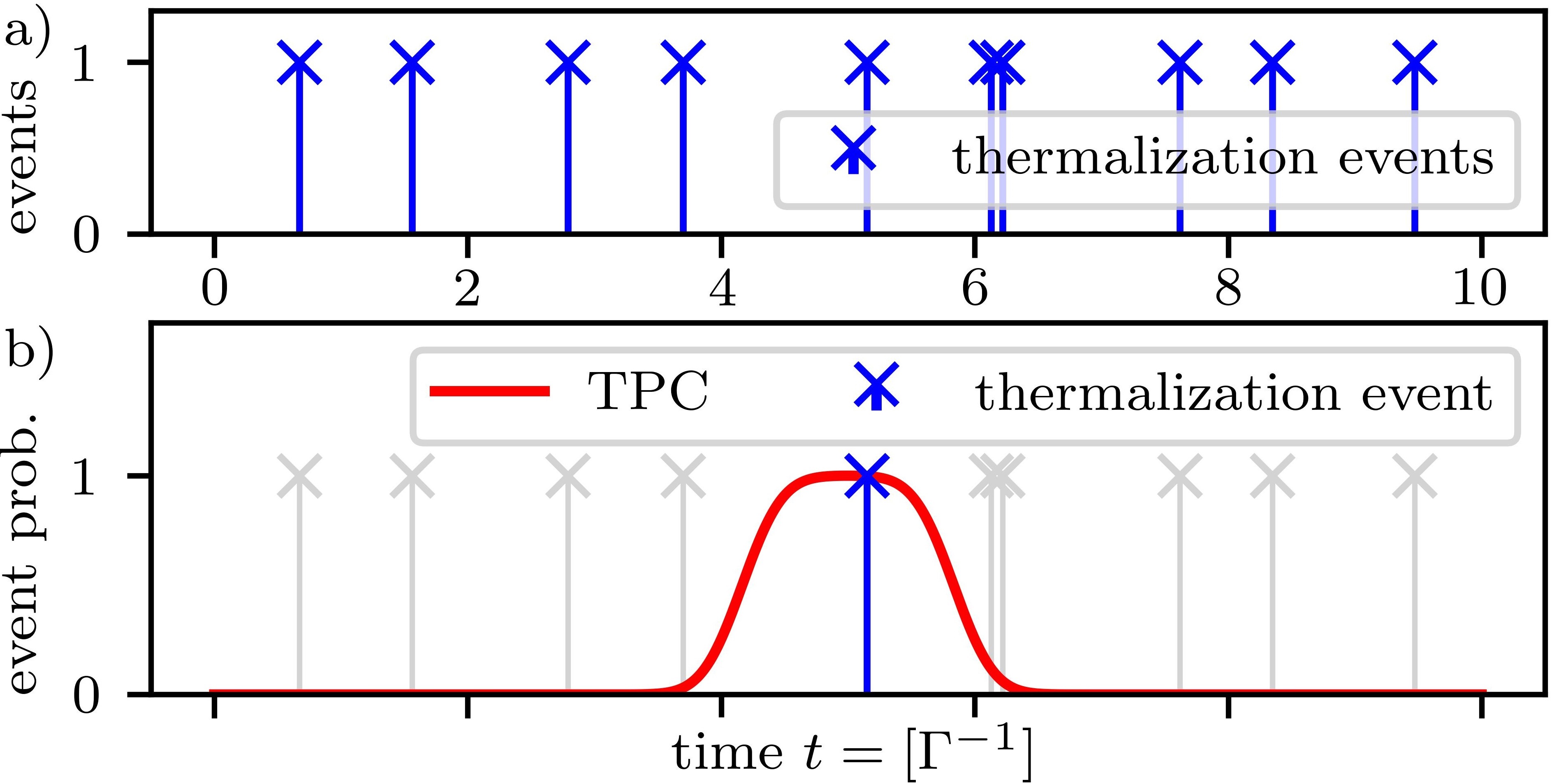}
    \caption{The two plots qualitatively show the ticking behavior of two clocks with respect to parameter time $t$ (horizontal axis). The ticks of such a clock are generated by individual thermalization events at rate $\Gamma$ (vertical stripes). Figure a) sketches a generic example, where these events are Poisson distributed. In b), temporal probability concentration is shown, another process through which the probability of an thermalization event can be suppressed at times (shown by the hat-shaped curve). As a result, the probability density for a tick can be concentrated around a desired average value, here $5\,\Gamma^{-1}$, with the tick time uncertainty of order $\Gamma^{-1}$ bounded by the width of the TPC window.}
    \label{fig:TPC_good_sampling}
\end{figure}
\begin{figure}
    \centering
    \includegraphics[width=\columnwidth]{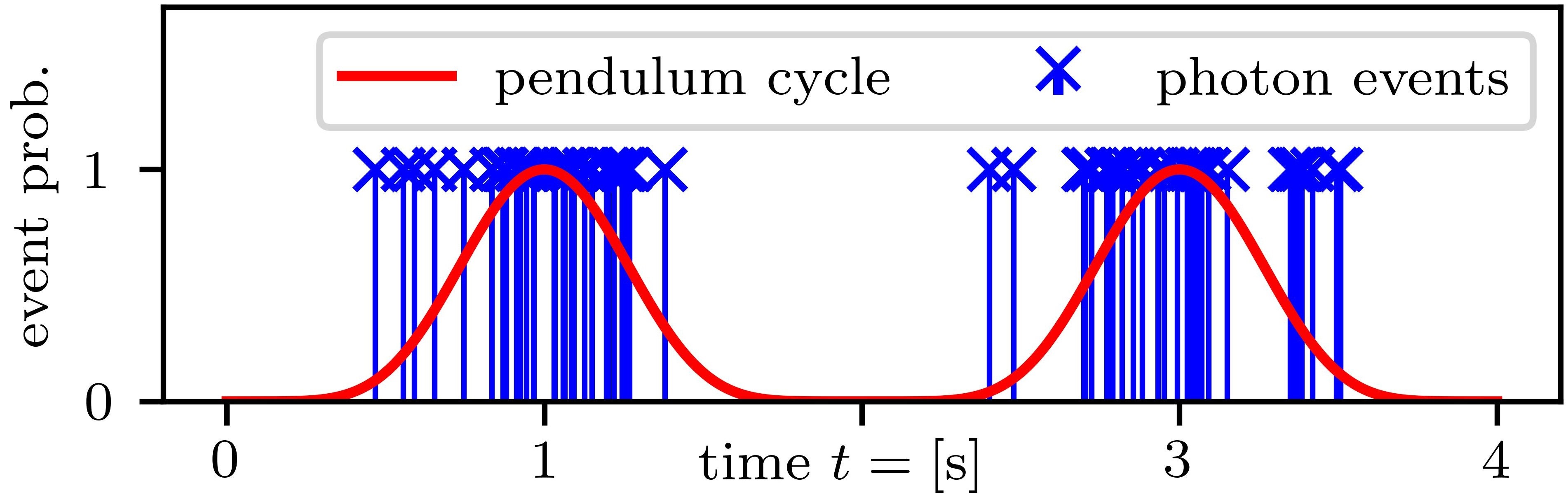}
    \caption{We illustrate the oversampling regime with a pendulum in a weakly lit environment. The two sources of entropy production for this clock are, firstly, the friction within the clockwork itself and, secondly, the light-matter interaction to track the position of the pendulum.
    While the entropy from friction can in principle be vanishingly small, the one from observation is fundamental and can not be made zero without losing the ability to measure time.
    The plot shows the elementary ticking events of this clock as a function of time, i.e., the photons reflected off the pendulum when it is close to its maximum deflection, and the pendulum oscillation (TPC). In the oversampling regime, the photon rate is much higher than the frequency of the TPC, and this allows to use the TPC period to define ticks.}
    \label{fig:TPC_oversampling}
\end{figure}

Examples of clocks following this scheme are given in \cite{Erker2017,Woods2019,Schwarzhans2021,Woods2021,Woods2022,Erker2014,Dost2023}, with a particularly illustrative one given in the first mentioned reference: a three-qubit system where one of the qubits couples dissipatively to the electromagnetic field with strength $\Gamma$, emitting photons when it decays, and these events are then counted as ticks.
Alone, the qubit would undergo exponential decay and give accuracy $N=1,$ but here it is autonomously driven by the other two qubits, which themselves are coupled to out-of-equilibrium heat baths.
The Hamiltonian coupling the three qubits is the periodic process which makes the effective decay probability of this clock time-dependent, and leads to an enhanced accuracy $N>1$ due to TPC. More details are in Sec.~\ref{sec:an_example} of the Appendix.
Many macroscopically sized clocks, be it pendulums or atomic clocks, operate in an oversampling regime, where multiple irreversible events occur per TPC cycle (see Figure \ref{fig:TPC_oversampling}). In this regime, the thermodynamic cost of a clockwork often becomes obscure, as both the dissipation due to the macroscopic number of irreversible events and the TPC have to be accounted for.
Atomic clocks, for example, do not count photons as a way to tell time, rather they use the oversampled coherent oscillations of a maser tuned some stable reference atomic transition to estimate the TPC frequency.
The time-scale $\Gamma$ of the fundamental ticking events appearing in eq.~\eqref{eq:nr_bound} is therefore not the limiting factor to the accuracy of atomic clocks, the stability of the coherent oscillation of the electromagnetic field is, i.e., the TPC stability. 
In atomic clocks, accuracy is examined using Allan Variance which captures the stability of the TPC oscillation over many different time-scales \cite{Allan1966,Allan1997,Riley2008}. 
Quantum projection noise, thermal noise but also natural drifts in the experimental setup are what affect the stability of atomic clocks \cite{Shiga2014,Schulte2020}.
A summary with some key references on the working principle of atomic clocks can be found in Sec.~\ref{sec:working_principle_atomic_clocks} of the Appendix.

So far, we have established that every clock is subject to irreversible processes and that through TPC, they can increase their accuracy.
In the following, we introduce a mathematical model to describe clocks on a quantum scale and where the two contributions (a) irreversible ticks and (b) temporal probability concentrations have an explicit representation in the equations of motion. 
Eventually, this framework allows us to formulate the fundamental trade-off between the accuracy and resolution of clocks, dictated by thermodynamics.
\paragraph*{Model.}
Quantum clocks \cite{Erker2017,Woods2019,Schwarzhans2021,Woods2021,Woods2022,Erker2014,Dost2023} only weakly coupled to a memoryless environment (sufficiently large thermal baths are one such environment) can be described by a Lindblad master-equation \cite{Lindblad1976}. If we are again talking about ticking clocks, their state can be Fourier-decomposed by introducing a free counting field $\chi$ \cite{Schaller2014,Bruderer2014,Landi2023},
\begin{align}
    \label{eq:classical_tick_mixture}
    \rho(t,\chi)=\sum_n \rho^{(n)}(t)e^{in\chi}.
\end{align}
Each non-normalized density matrix $\rho^{(n)}(t)$ can be thought of as the system's state conditioned on $n$ ticks having already occurred.
In a setting where the interactions with the environment are memoryless, the ticks are produced by linear jump operators $J_j$, the generators of the process (a). Aside from this, the clock is subject to a general open quantum system's evolution with Lindblad operator $\mathcal{L}_\text{no tick}$ which is the generator of TPC, i.e., process (b).
We are interested in the statistics of the time $T$ between any two successive ticks $n$ and $n+1$.
These statistics may differ from each tick to the next one because the initial state changes, and this also means that the accuracy $N$ and frequency $\nu$ can change with each tick $n$. However, the trade-off theorem is agnostic to the value $n$, and holds for all pairs of $N$ and $\nu.$
Without loss of generality, we can therefore assume that the $n$th tick happened at time $t=0$ and that $n=0$, and we can look at the evolution of the conditional state $\rho^{(0)}(t)$ that exactly $n=0$ ticks have occurred.
Given the initial state $\rho^{(0)}(0)$ of the clock from eq.~\eqref{eq:classical_tick_mixture}, the evolution is entirely determined by
\begin{align}\label{eq:no_tick_state_evolution}
    \dot{\rho}^{(0)}(t)=\mathcal{L}_\text{no tick}[\rho^{(0)}(t)]-\frac{1}{2}\sum_j\left\{J_j^\dagger J_j,\rho^{(0)}(t)\right\}.
\end{align}
The right-most term of eq.~\eqref{eq:no_tick_state_evolution} produces the ticks and is therefore responsible for the process (a) while $\mathcal{L}_\text{no tick}$ generates the TPC, i.e., process (b).
In this sense, one may attempt to separate (a) and (b) into two independent processes, as in \cite{Schwarzhans2021}; however the back-action of the tick channel always affects the clock evolution which is an inherent feature of clocks that operate on a quantum scale.

\paragraph*{Ticking statistics.} 
Together with the initial state, eq.~\eqref{eq:no_tick_state_evolution} defines the evolved state $\rho^{(0)}(t)$ and therefore entirely determines how the tick time random variable $T$ is distributed, because the cumulative probability $P[t\leq  T]$ that no tick has occurred up to time $t$ equals the trace $P[t\leq T]=\tr\rho^{(0)}$.
To highlight the influence of TPC, we can equivalently express the cumulative tick probability as a function of the conditional tick rate $P[T=t|T\geq t]$, defined as the instantaneous tick probability at time $t$ conditioned on the tick not having happened before. We can obtain this rate by working with the clock's state conditioned on not having decayed $\rho^\text{no tick}=\rho^{(0)}/\tr\rho^{(0)}$.
Then, the conditional tick rate equals the trace $\tr(V\rho^\text{no tick}(t')),$ where $V=\sum_j J_j^\dagger J_j$ is the positive operator generating the clock's ticks \cite{Meier2022}, and we find the relationship
\begin{align}\label{eq:cumulative_non_tick}
    P[t\leq T]=\exp(-\int_0^t dt' \tr(V\rho^\text{no tick}(t'))),
\end{align}
which we further outline in the Appendix~\ref{sec:general_formalism}.

The master-equation description of quantum ticking clocks allows us to formalize the statement made in the accuracy-resolution trade-off theorem. The tick channel is governed by the positive operator $V$, whose spectral decomposition reveals the time-scales involved in the decay. We define the rate $\Gamma$ as the fastest one of them,
\begin{align}
\label{eq:def_gamma}
    \Gamma := \|V\|_\text{max}=\max_{\rho \in \mathcal{S}(\mathcal{H}_C)}\tr V\rho,
\end{align}
where the maximum is taken over all possible clock states $\rho\in\mathcal{S}(\mathcal{H}_C)$.
This is consistent with the special case of exponential decay (as used in \cite{Barato2016,Erker2017}) with rate $\Gamma,$ since there, the tick generator is given by a single jump operator $J=\sqrt{\Gamma}\ketbra{0}{1}$.
If a clock is described by eq.~\eqref{eq:no_tick_state_evolution} and produces its ticks by means of the generators $J_j,$ then the accuracy is limited by the resolution, regardless of how well a possible clockwork in the background works. The resulting bound is eq.~\eqref{eq:nr_bound} from the accuracy-resolution trade-off theorem which we recall here for completeness,
\begin{align*}
    N\leq \frac{\Gamma^2}{\nu^2},
\end{align*}
and prove in the following.
\begin{proof}
The key observation is that no clock can on average tick faster than the decay process that mediates the ticks allows for. If the elementary ticks are generated by an ensemble of jump operators $J_j$, then the fastest such rate is given by $\Gamma$ (see eq.~\eqref{eq:def_gamma}).
A single such channel produces ticks that are exponentially distributed and as a consequence, we can reduce the generic form from eq.~\eqref{eq:cumulative_non_tick} to the special case of exponential decay, where the exponent reduces to $-\Gamma \int dt' p(t'),$ and $p(t)$ is a measure of the clock's state population that can decay.
Regardless of how the TPC modulates $p(t),$ the variance $\sigma^2$ of the tick can never be smaller than $\Gamma^{-2},$ the variance of exponential decay. This is a manifestation of the fact that no clock can tick faster than it's underlying decay process. Now that the variance $\sigma^2$ of the tick time distribution is bounded from below by $\Gamma^{-2}$, the main theorem follows. We refer the reader to Sec.~\ref{sec:proofs} of the Appendix for a detailed account of the proof. 
\end{proof}

\begin{figure}
    \centering
    \includegraphics[width=\columnwidth]{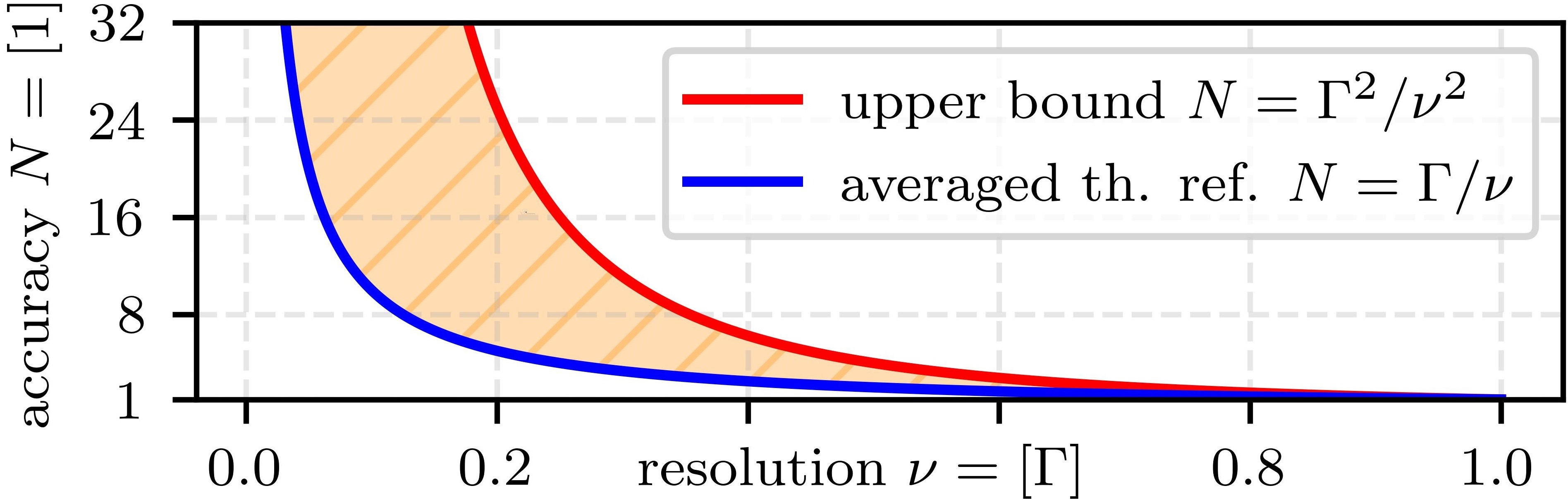}
    \caption{In this figure, we plot the accuracy $N$ on the vertical axis versus the resolution $\nu$ on the horizontal axis. The red curve highlights the upper bound for the accuracy $N=\Gamma^2/\nu^2$ given by the accuracy-resolution trade-off theorem. 
    Below, we have the blue curve, indicating the accuracy and resolution that can be achieved by averaging the thermal reference clock with fundamental tick rate $\Gamma$.
    The shaded area between the two curves contains all the accuracy-resolution touples that cannot be achieved by classically averaging a rate $\Gamma$ stochastic process, but which are still allowed by the upper bound. How close one can get to the red curve with quantum clocks and whether classical clocks are constrained to be below the blue line are open questions.}
    \label{fig:nr_plane}
\end{figure}
In this letter, we have analyzed quantum timekeeping devices through a thermodynamic lens, where their tick generation can be decomposed into two processes, (a) a stochastic process which irreversibly produces the ticks and (b) temporal probability concentration (TPC) through a clockwork which controls when these elementary ticks occur.
Then, we asserted that there is a fundamental trade-off between the clock's accuracy and resolution (see trade-off theorem but also Fig.~\ref{fig:nr_plane}), stating that the number of times a clock can tick until it goes wrong by one tick is universally bounded by the inverse of the clock's resolution squared.

\paragraph*{Atomic and optical clocks.}
Timekeeping devices fundamentally require stochastic thermalization events to measure time, and we have shown that for the class of clocks using those events directly to define ticks, the trade-off from eq.~\eqref{eq:nr_bound} applies.
Macroscopic clocks such as atomic and optical ones work in a different regime, where the TPC is sampled by irreversible events and ticks are defined not by event numbers directly, but rather those are used to estimate the oscillatory TPC process frequency.
This frequency is usually downsampled to a lower frequency, like for example the $10\,\mathrm{MHz}$, or the $1\,\mathrm{Hz}$ standard that is then used to generate an electrical signal to read out the ticks at said resolution \cite{Marlow2021}.
Our work points out some of the challenges that have to be overcome for optical clocks being used to produce ticks at $\mathrm{THz}$ resolution: for one, ultra-fast electronics that can generate an electric signal using only a single photon per oscillation of the e.m.\ field. For another, this gives an estimate of the power required from a laser to create such a photon flux.
While for time-standards whose main goal it is to provide long-time stability, power consumption may not be the primary concern, there are other uses for clocks where energy-efficiency matters.
For example, quantum technologies require accurate high-resolution timers that do not disturb the fragile state of the quantum system through heat dissipation \cite{Ball2016,Buffoni2022,Tholen2022,Taranto2023}.
For building noise-robust quantum devices, it may thus become unavoidable to account for the thermodynamic resources consumed by clocks, which is where we expect quantum clocks to outperform their macroscopic classical counterparts.

\paragraph*{Achievability.} 
Good clocks excel by resolving time well while at the same time being highly accurate, i.e., maximizing both $\nu$ and $N$; ideally, they do so at optimal thermodynamic costs far away from the oversampled regime and only with a single thermalization event per tick. As a foundational question, we may ask, does there exist an appropriate clock Lindblad operator $\mathcal{L}$, through which it is possible to saturate the bound $N=\Gamma^2/\nu^2$? And the answer is: No, at least not for finite systems.
As it turns out, saturating this bound amounts to a time-dependent Lindblad-operator, which instantaneously rotates the clock state from a subspace of states where it can't tick onto a state from which it decays at rate $\Gamma$ (see proof in Sec.~\ref{sec:proofs} of the Appendix).
The TPC in this case would be ideal as it concentrates the decay event to the most narrow time-window possible, the one given by the underlying stochastic process, but quantum speed limits prohibit such an instantaneous state rotation for systems finite in energy and dimension \cite{Mandelstam1991,Margolus1998,Deffner2017}.
A weaker question we can ask is whether at least the scaling $N=O(\Gamma^2/\nu^2)$ can be reached, and which resources are required for this.
Without imposing any restriction on the clock's Lindbladian aside from time-independence, there are quantum clocks which asymptotically reach the squared scaling of $N$ for large dimensions of the clock Hilbert space \cite{Woods2019}. To do so, they require highly coherent states, whose generation using only thermal resources is technically infinitely expensive from an entropic perspective \cite{Taranto2023} and the approximation with finite resources is an open problem.
Of interest in the field of thermodynamics are the autonomous quantum clocks which only require thermal resources to run \cite{Erker2017,Schwarzhans2021,Manikandan2022,Meier2022, Erker2014}, in particular no external control but also no coherence in the initial state. It is ongoing research, whether it is possible for such clocks to approach the optimal accuracy-resolution scaling.

Questions about fundamental precision limits are generally of great interest. The accuracy of clocks falls into this category and closely related problems have been examined in the field of (quantum) stochastic thermodynamics under the name of thermodynamic uncertainty relations (TUR) \cite{Barato2015,Horowitz2020,Vu2022} and kinematic uncertainty relations (KUR) \cite{Terlizzi2018,Vo2022,Prech2023}. Future work exploring connections between timekeeping and the TUR / KUR has to reveal whether a quantum thermodynamic advantage close to the optimal accuracy-resolution bound $N=\Gamma^2 / \nu^2$ is in principle possible. 
Both superconducting circuits and optomechanical systems are promising platforms to test the achievability of the optimal accuracy-resolution relation while accounting for the resources and ensuring that they do not introduce a hidden clock through a backdoor.

\paragraph*{Acknowledgements.} We wish to acknowledge fruitful discussions and feedback from Ralph Silva, Nuriya Nurgalieva, Renato Renner, Maximilian Edward Lock and Jake Xuereb.
F.M., P.E. and M.H. acknowledge the financial support from the European Commission under the Horizon Europe project ASPECTS (Grant No. 101080167) and the ERC consolidator grant COCOQUEST (Grant No. 101043705).
F.M. acknowledges the SEMP scholarship from Movetia for his research stay at Atominstitut, Technische Universität Wien.
E.S. acknowledges the support from the Austrian Science Fund (FWF) through the START project Y879-N27, the ESQ Discovery grant “Emergence of physical laws: From mathematical foundations to applications in many body physics” and the ERC consolidator grant COCOQUEST (Grant No. 101043705).
P.E. and M.H. further acknoledge funds from the FQXi (FQXi-IAF19-03-S2)
 within the project ``Fueling quantum field machines with information''.

%

\clearpage

\appendix
\section*{Appendices}
\section{Temporal probability concentration}
\label{sec:temporal_probability_concentration}
In this section, we elaborate on the technical details regarding the concept of temporal probability concentration (TPC) which we have introduced in the main text of this letter.
This concept has first been explicitly mentioned in~\cite{Schwarzhans2021}, but implicitly, it has been applied in a variety of works~\cite{Erker2017,Woods2019,Woods2022,Dost2023}.
In the main text, a distinction of tick generation into two separate processes was made: a) the irreversible process whose occurrence \textit{defines} the tick, and b) an additional filter mechanism, which temporally concentrates the probability of when a tick, i.e., process a) occurs. We refer to the latter as TPC.

\subsection{An example}
\label{sec:an_example}
As an illustration, let us take a fixed exponential decay process, i.e., a physical system comprising two states, say an excited one and the ground state, and the excited state decays with a fixed rate $\Gamma$ into the ground state.
If we were to initialize this system in the excited state, its tick rate would be constantly equals to $\Gamma$ and the resulting tick probability density $P[T=t]=\Gamma e^{-\Gamma t}$ would be exponential.
Not unsurprisingly, this clock does not excell in accuracy, which manifests itself in the accuracy $N=1,$ as already elaborated in the main text.
Let us now introduce an additional dynamical process, for example the driving by an autonomous thermal machine as in~\cite{Erker2017}. Such an interaction could first rotate the clock state from the ground state from which it can not decay into the excited state, from which it can decay, leading to an effective decay rate not always equals $\Gamma$, but now modulated dynamically by the autonomous thermal machine. This time-dependency is the TPC of the tick process, and in the example of~\cite{Erker2017} it leads to an accuracy $N>1.$

We here summarize the clock example from~\cite{Erker2017} to illustrate the general model for clocks introduced with the equations of motion in eq.~\eqref{eq:no_tick_state_evolution} in the main text.
In general, the two-level system in the clock described above can be extended to a $d$-dimensional ladder which allows to delay the decay process even further.
To be more specific, the clock model comprises the ladder $\mathcal H_L$ spanned by the states $\ket{0}_L, \ket{1}_L,\dots,\ket{d-1}_L$, and the thermal machine. The thermal machine itself is made up of two qubits, a \textit{cold} one $\mathcal H_C$ with the states $\ket{0}_C,\ket{1}_C$ and a \textit{hot} one $\mathcal H_H$ with states $\ket{0}_H,\ket{1}_H.$
The interactions between the ladder and the two qubits is described by the Hamiltonian $H=H_0 + H_\text{int},$ where 
\begin{align}
    H_0 = \omega_C \ketbra{1}{1}_C + \omega_H\ketbra{1}{1}_H + \sum_{n=0}^{d-1} n\omega_L\ketbra{n}{n}_L,
\end{align}
is the system's free Hamiltonian, with $\omega_C$ the energy-splitting of the cold qubit, $\omega_H$ that of the hot qubit, and $\omega_L$ that of the ladder. The term
\begin{align}
    H_\text{int} = g\sum_{n=0}^{d-1} \Big(\ketbra{10}{01}_{CH}\otimes\ketbra{n+1}{n}_L + \mathrm{h.c.}\Big)
\end{align}
describes the population exchange between the two-qubit thermal machine and the ladder system. 
Furthermore, there are the thermal dissipators $\mathcal L_C,$ and $\mathcal L_H$ modelling the interactions of the cold qubit with a cold bath at inverse temperature $\beta_C$ and those of the hot qubit with it's hot bath at inverse temperature $\beta_H$.
The dissipators are of the form $\mathcal L_C = n_C \gamma_C\mathcal D[\ketbra{1}{0}_C]+(1+n_C)\gamma_C\mathcal D[\ketbra{0}{1}_C]$, where $n_C$ is the photon number of the cold bath given by Bose-Einstein statistics, $n_C = (e^{\beta_C\omega_C}-1)^{-1},$ and $\gamma_C$ is the coupling rate of the cold bath to the cold qubit.
The dissipator for the hot bath is analogous but with the subscript $H$ instead of $C$.
The terms $\mathcal D$ are as usually defined as $\mathcal D[L] = L\circ L^\dagger - \frac{1}{2}\left\{L^\dagger L,\circ\right\}.$
All these expressions together give rise to the Lindbladian part $\mathcal L_\text{no tick}$ responsible for TPC,
\begin{align}
    \mathcal L_\text{no tick} = -i[H,\circ] + \mathcal L_C + \mathcal L_H.
\end{align}
The tick generation then comes from the ladder decaying from it's top level state into the ground state through the tick generating operator $J, $ defined by
\begin{align}
    J = \sqrt{\Gamma} \ketbra{0}{d-1}_L,
\end{align}
where $\Gamma$ is the coupling strength of the ladder's transition $\ket{d-1}_L \rightarrow\ket{0}_L$ to the environment, i.e., the decay rate.
If we go through the calculations provided in reference~\cite{Erker2017}, we find that this clock's accuracy in the regime $\Gamma, g\ll \gamma_H,\gamma_C$, and $g\ll \omega_C,\omega_L$, and under the resonance condition $\omega_C + \omega_L = \omega_H$ grows with the dimension $d$ but the resolution decreases.
In the asymptotic limit, as cold bath temperatures goes to absolute zero, $\beta_C^{-1} \rightarrow 0,$ the accuracy grows linearly in the dimension, $N\propto d,$ and the resolution inverse linearly, $\nu\propto d^{-1}.$
The resulting accuracy-resolution relationship is also inverse linear for this clock, $N \propto \nu^{-1}$; as the authors in~\cite{Erker2017} state, the clock in this regime behaves essentially classically, hence, not unexpectedly, we find the classical trade-off from eq.~\eqref{eq:ccg_bound} in the main text.

\subsection{General formalism}
\label{sec:general_formalism}
In the following we approach the concept of TPC from two different perspectives, firstly from the equations of motion by starting from eq.~\eqref{eq:no_tick_state_evolution}, describing explicitly the evolution of a given tick state, and secondly from a formal probability theoretic paradigm.

\paragraph*{TPC as an emergent property from the equations of motion.}
For convenience, let us recall the equation of motion~\eqref{eq:no_tick_state_evolution} from the main text,
\begin{align}
\label{eq:no_tick_state_evolution_appendix}
    \dot{\rho}^{(0)}(t)=\mathcal{L}_\text{no tick}[\rho^{(0)}(t)]-\frac{1}{2}\sum_j\left\{J_j^\dagger J_j,\rho^{(0)}(t)\right\}.
\end{align}
Ticks, by definition, are the transitions generated by the operators $J_j$, and this is the process a) in our list from before. This means, if the clock system undergoes any one of the stochastic transitions $J_j,$ this is counted as a tick.

We can take the trace of the equation of motion~\eqref{eq:no_tick_state_evolution_appendix}, and we obtain on the left-hand side the derivative of the cumulative tick probability density,
\begin{align}
    \tr\left\{\dot\rho^{(0)}(t)\right\} = \frac{d}{dt}P[t\leq T].
\end{align}
For a well-behaved, here, continuously differentiable, cumulative probability $P[t\leq T]$, the derivative equals the probability density function (PDF) $P[T=t]$. This allows us to identify the trace of the right-hand side of eq.~\eqref{eq:no_tick_state_evolution_appendix} as the probability density,
\begin{align}
    P[t=T] = \tr\left\{ V \rho^{(0)}(t)\right\},
\end{align}
where $V=\sum_j J_j^\dagger J_j$ as defined in the main text. This probability distribution completely describes the statistics of the tick considered. Average time between ticks $\mu$ and variance of the time between ticks $\sigma^2$ are defined with respect to this PDF.
An related quantity, which is of particular relevance from an operational perspective is the conditional tick probability density: the probability density $P[T=t | T\geq t]$ that the tick occurs at time $t$ conditioned on the fact that it has not ticked before that time $t$. We can calculate the conditional PDF using
\begin{align}
    \label{eq:cond_tickPDF}
    P[T=t | T\geq t] &= \frac{P[T=t \wedge T\geq t]}{P[T\geq t]} \\
    &=\frac{P[T=t]}{P[T\geq t]},
\end{align}
which reveals that the conditional tick PDF equals the trace that appeares in eq.~\eqref{eq:cumulative_non_tick} in the main text,
\begin{align}
    P[T=t | T\geq t]=\tr(V\rho^\text{no tick}(t')).
\end{align}
In this form, TPC becomes particularly apparent: for a clock without any internal dynamics, i.e., $\mathcal L_\text{no tick} \equiv 0$ where ticks are generated through exponential decay, the conditional tick PDF is constant, $P[T=t]=\Gamma,$ where $\Gamma$ is the rate of the exponential decay.
This is the example, where the stochastic thermalization event that generates the tick is not temporally concentrated, hence the time-independence of the conditional tick PDF.
In contrast, for a non-trivial clockwork as for example presented in the references~\cite{Erker2017,Woods2019,Woods2022,Dost2023}, the conditional tick PDF has time-dependency, which eventually leads to a non-exponential tick probability with accuracy $N>1.$

\paragraph*{TPC as a conditional probability density.}
The previous paragraph introduced the notion of the conditional tick PDF $P[T=t|T\geq t]$ that a tick happens at time $t$ conditioned on the tick not having happened before.
Equation~\eqref{eq:cumulative_non_tick} in the main text relates this expression to the tick cumulative tick probability density $P[T\geq t]$ via the more general identity,
\begin{align}\label{eq:cond_tickPDF_exponential}
    P[T\geq t] = \exp\left({-\int_0^t d\tau P[T=t|T\geq t]}\right).
\end{align}
Here we would like to derive this expression and elaborate. 
Definition of cumulative probability distributions (CDFs) and PDFs ensure, that
\begin{align}
    \frac{d}{dt} P[T\geq t] &= -P[T=t],
\end{align}
so long as the derivative is well-defined and continuous.
If we then use the conditional probability law from eq.~\eqref{eq:cond_tickPDF} of this supplemental material, we find that
\begin{align}
    \frac{d}{dt} P[T\geq t] &= -P[T\geq t]P[T=t|T\geq t].
\end{align}
The solution of this differential equation is the exponential expression as generally given in eq.~\eqref{eq:cond_tickPDF_exponential}.
This shows the relation in eq.~\eqref{eq:cumulative_non_tick}.

\section{Working principle of atomic clocks}
\label{sec:working_principle_atomic_clocks}
Atomic clocks are the technological state-of-the-art when it comes to timekeeping and in the letter, we discuss the relevance of the trade-off theorem to atomic clocks.
In this section, we give the interested reader a brief description of the working principle of atomic clocks and refer them to more specialized references.

\paragraph*{Frequency estimation vs. tick generation.}
Atomic clocks use the period of the coherent electromagnetic field oscillation in a laser to measure time and they use a reference frequency, for example electronic transitions in an atom, to stabilize the laser's frequency.
Commonly, atomic clocks are used to generate a stable frequency reference. The coherent light field of the laser which is the basis of every atomic clock is generated using stimulated emmission (see e.g.~\cite{Carmichael1993} for an open quantum system's approach or~\cite{Eichhorn2014}, for an applied textbook).
In practice, the laser frequency is inevitably subject to  drifts and noise due to interactions with the environment and inherent imperfectness of the constituents (e.g.~Brownian motion). To correct for this, elaborate feedback techniques have been developed, e.g., the hyperfine groundstate transition of an Caesium-133 atom is used as a frequency reference to stabilize the laser~\cite{Essen1955,Major2007,Bize2019,Marlow2021}.
In our model introduced in the main text, all frequencies are idealized to be perfectly stable; in particular also the TPC's oscillatory frequency is assumed to be unchanging.
We reveal that despite these idealizations, such a clock can not be perfect, due to the underlying stochasticity of the thermal processes. Formally, we show this with the accuracy-resolution trade-off theorem.

However, frequency stabilization is only part of the story required for accurate timekeeping, as an additional process is required to produce ticks.
In the oversampling regime atomic clocks work in, ticks are not uniquely defined but one may chose to define ticks as the zero-crossings or the maxima of the oscillation.
In this regime, the elementary stochastic events (here: photoemission) are not directly used as a tick definition anymore but rather as a means to sample the TPC oscillation, which then defines ticks as for example one period of the oscillation.
Only in the regime where such a clock is run with a single photon per oscillation of the electromagnetic field, and the photon defines a tick, does the trade-off theorem impose a practical restriction on the clock's performance.

\section{Proof of the accuracy-resolution trade-off}
\label{sec:proofs}
Before we get started with the proof of the accuracy-resolution trade-off theorem (which is a generalization of the one shown in \cite{Meier2022}), some preliminary notation has to be established. When it comes to the cumulative non-tick probability $P[t<T]$ as in eq.~\eqref{eq:cumulative_non_tick} of the main text, the trace $\tr V \rho_C(t)$ can be rewritten as
\begin{align}
    \tr( V \rho_C(t)) = \Gamma p(t),
\end{align}
where $p : \mathbb{R}_{\geq0}\rightarrow[0,1]$ is a smooth function that takes values between $0$ and $1$. This comes from the fact that $\Gamma$ is the maximum of the expression over all states $\rho$ and the function can only take positive values because $V$ is a positive operator. This expression puts is into the position, where we can write
\begin{align}
\label{eq:non_tick_prob_general_decay}
    P[t\leq T]=\exp(-\Gamma\int_0^t dt' p(t')),
\end{align}
which is formally the exact same expression that one would obtain in the case where the clock's elementary ticks are exponential decay. We will from now on refer to $p(t)$ as the \textit{top-level population}. The remaining sections are an adapted version of the proof of the accuracy-resolution theorem originally derived in \cite{Meier2022}.

\subsection{The Heaviside-$\Theta$ population}
\label{sec:heaviside_population}
\begin{restatable}[Heaviside populations]{definition}{heavisidefamily}
\label{def:heavisidefamily}
The family of Heaviside populations comprises all the functions
\begin{align}
    p^\Theta : \mathbb{R}_{\geq 0}&\rightarrow [0,1]\\
    t&\mapsto \Theta (t-t_0),
\end{align}
where $t_0\geq 0.$
\end{restatable}
\begin{figure}[b]
    \centering
    \includegraphics[width=\columnwidth]{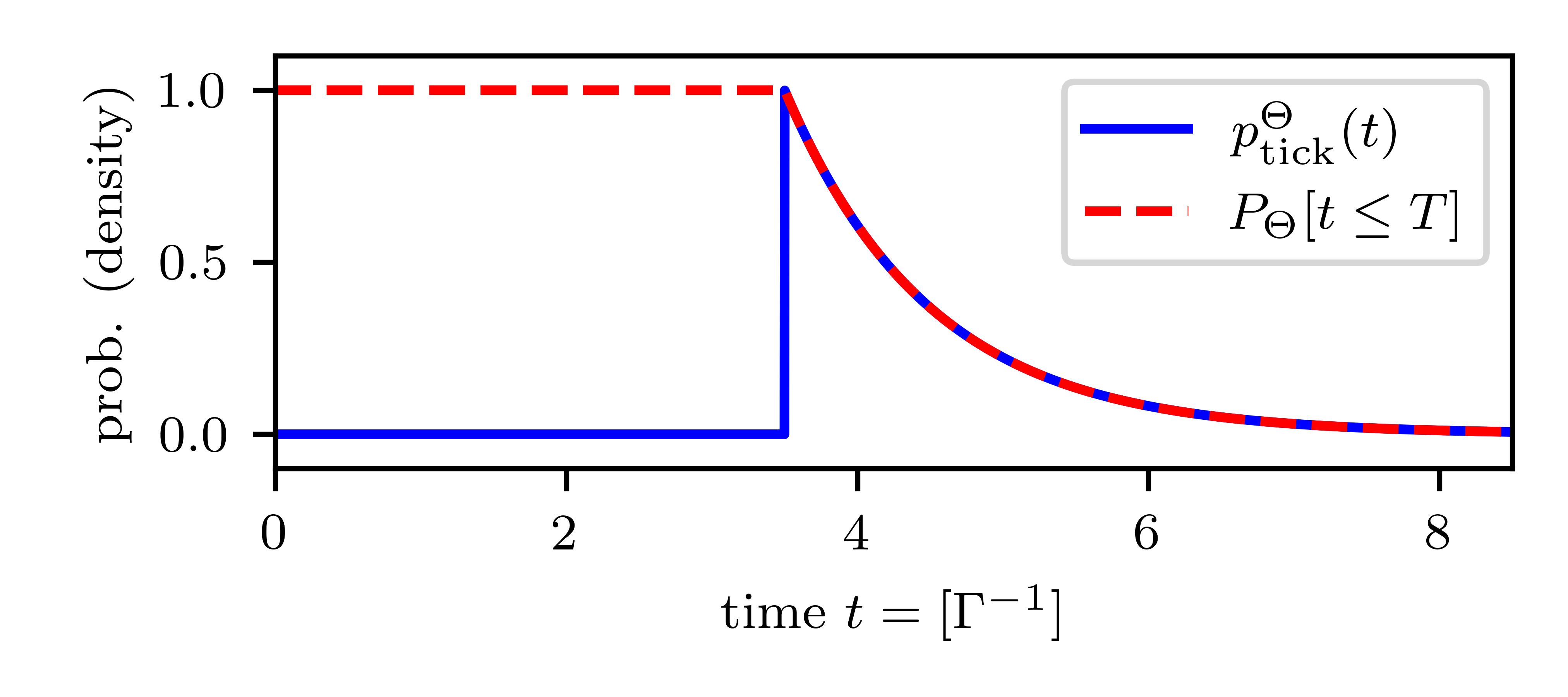}
    \caption{This graph visualizes the ticking statistics for the Heaviside population. The solid line in the above figure visualizes the tick probability density, whereas the dashed line stands for the cumulative non-tick probability. The time axis scales in inverse units of the decay rate $\Gamma$ and the population in the example is chosen such that on average the decay occurs at $\mu=4.5 \,\Gamma$.}
    \label{fig:heaviside_family}
\end{figure}
\paragraph*{Properties of the Heaviside population.}
Such a top-level population is unphysical because of its discontinuity at time $t=t_0.$ This is not an obstacle, though, as we only use the Heaviside population to prove an upper bound for the accuracy of a decay clock. No claim is made whether we can physically obtain the Heaviside population.
The non-tick probability $P_\Theta[t\leq T]$ and tick probability density $p_\text{tick}^\Theta(t)$ associated to this population are given by (see eq.~\eqref{eq:non_tick_prob_general_decay})
\begin{align}
    P_\Theta[t\leq T] =
    \begin{cases}
        1 & t\leq t_0 \\
        e^{-\Gamma(t-t_0)} & t\geq t_0,
    \end{cases}
\end{align}
and
\begin{align}
    p_\text{tick}^\Theta(t) =
    \begin{cases}
        0 & t\leq t_0 \\
        \Gamma e^{-\Gamma(t-t_0)} & t\geq t_0.
    \end{cases}
\end{align}
See Fig.~\ref{fig:heaviside_family} for a visualization of these functions.
Using the analytic solution of the integral over an exponential, $\int_0^\infty dx e^{-ax} = 1/a,$ we can calculate the accuracy and resolution for clocks working the Heaviside population.
\begin{restatable}[]{lemma}{heavisideaccuracyresolution}
\label{lemma:heavisideaccuracyresolution}
The Heaviside population $p^\Theta(t)=\Theta(t-t_0)$ has an accuracy and resolution given by
\begin{align}
\label{eq:resolution_heaviside}
    N=\left(1+\Gamma t_0\right)^2,\text{ and }\,\nu=\frac{1}{t_0+\frac{1}{\Gamma}},
\end{align}
leading to an accuracy-resolution relation
\begin{align}
    N=\frac{\Gamma^2}{\nu^2}.
\end{align}
\end{restatable}
Before we prove Lemma~\ref{lemma:heavisideaccuracyresolution}, let us put the Heaviside population into perspective. Given the accuracy-resolution trade-off theorem the equality $N=\Gamma^2/\nu^2$ for clocks with Heaviside population tells us that this top-level population profile is ideal in the accuracy-resolution sense. This is no coincidence: a Heaviside population with offset at $t_0$ requires a clock in the background which perfectly knows $t_0$. The population then describes a ladder conditioned on not having decayed, whose top-level is populated precisely at $t=t_0,$ and stays there. This would require instantaneous external driving at $t=t_0$, i.e., a perfect background clock. The exponential decay channel coupled to the top-level then smears out the Heaviside profile of $p^\Theta(t)$ into an exponential decaying tick probability density $p_\text{tick}^\Theta(t)$ of width $1/\Gamma,$ giving rise to the results of Lemma~\ref{lemma:heavisideaccuracyresolution}.
\begin{proof}
The average tick time is given by an integral over $t p_\text{tick}^\Theta(t)$ and splits into two parts
\begin{align}
    \mu=\int_0^\infty dt\, t p_\text{tick}^\Theta(t) &= 0+\int_{t_0}^\infty dt\, t e^{-\Gamma(t-t_0)}\\
    &=t_0+\frac{1}{\Gamma}.
\end{align}
Thus, the resolution is given by $\nu=\left(t_0 + \frac{1}{\Gamma}\right)^{-1}$.
For the accuracy, which is defined as $N=(\mu/\sigma)^2,$ we only need to calculate the variance $\sigma^2.$ The variance, however, is invariant under translations of the tick probability density and, thus, one can calculate the variance without loss of generality for $t_0=0,$
\begin{align}
    \sigma^2 &= \int_0^\infty dt\, t^2 e^{-\Gamma t} -\mu^2\\
    &=\frac{2}{\Gamma^2}-\frac{1}{\Gamma^2}=\frac{1}{\Gamma^2}.\label{eq:variance_heaviside}
\end{align}
Expressing the accuracy in terms of $t_0,$ we find $N = (1+\Gamma t_0)^2$, that is, the accuracy increases with higher offset $t_0$. There is a tradeoff, though: The greater $t_0,$ the lower the resolution. Eliminating the $t_0$-dependency in the equations, we can establish an accuracy-resolution relation for the family of Heaviside populations given by $N=\Gamma^2/{\nu^2}.$
\end{proof}
\paragraph*{Analytic preliminaries.}
The tick probability density $p_\text{tick}(t)$ reveals how likely it is for a tick to happen during a given time interval. By the fundamental theorem of calculus, the cumulative non-tick probability $P[t\leq T]$ (whose negative derivative is $p_\text{tick}(t)$) contains the same information, and for proving the accuracy-resolution trade-off, the latter function turns out to be useful. Let us collect some general identities for the non-tick probability $P[t\leq T]$ which for some function $p(t)$ (be reminded, that the superscript \textit{no tick} is suppressed) is given by
\begin{align}
\label{eq:integral_cumulative}
    P[t\leq T]=e^{-\Gamma\int_0^t d\tau\,p(\tau)}.
\end{align}
There is a general constraint (Lemma~\ref{lemma:constraintcumulative}) on how fast the non-tick probability decays. The fact that $p(t)\in[0,1]$ ensures that $P[t\leq T]$ is monotonically decreasing (because the integral $\int_0^t d\tau\, p(\tau)$ is monotonically increasing) but does this not faster than exponentially.
\begin{restatable}{lemma}{constraintcumulative}
\label{lemma:constraintcumulative}
For all $t>0$ and $s>0$, the cumulative non-tick probability satisfies the following inequalities:
\begin{align}
    P[t\leq T]e^{-\Gamma s}\leq P[t+s\leq T]\leq P[t\leq T].
\end{align}
\end{restatable}
\begin{proof}
This is a consequence of eq.~\eqref{eq:integral_cumulative} and the fact that $0\leq p(t)\leq 1.$
\end{proof}
The cumulative non-tick probability can be used to calculate the moments of the probability density $p_\text{tick}(t)$ (see Definition \ref{def:moments} and Lemma~\ref{lemma:momentsidentity}) via the relation $p_\text{tick}(t)=-\dot{P}[t\leq T]$ and partial integration. It is important to note, that this discussion only makes sense for clocks that tick with certainty, i.e., $\lim_{t\rightarrow\infty}P[t\leq T]=0$. This need not be true for all clocks, however, whenever we talk about a \textit{tick probability density}, we implicitly assume that we have a properly normalized probability density in the probability theoretic sense \cite{Klenke2020}. 
\begin{restatable}[Moments]{definition}{moments}
\label{def:moments}
For a given tick probability density $p_\text{tick}(t)$, define its $k$-th moment as
\begin{align}
    t_k:=\int_0^\infty dt\, t^kp_\text{tick}(t).
\end{align}
\end{restatable}
\begin{restatable}[Tick probability moments]{lemma}{momentsidentity}
\label{lemma:momentsidentity}
    The $k$-th moment of $p_\text{tick}(t)$ is related to $P[t\leq T]$ by the integral
    \begin{align}
        t_k = k\int_0^\infty dt\, t^{k-1}P[t\leq T].
    \end{align}
\end{restatable}
\begin{proof}
This is partial integration and for the boundary conditions, we use the assumption that $\lim_{t\rightarrow\infty}P[t\leq T]=0$.
\end{proof}
In particular, we can apply this result to $\mu=t_1$ and $\sigma=t_2-\mu^2$, two expressions that are essential in Section \ref{sec:upper_bound_proof}, where we prove the upper bound. The average tick time $\mu=t_1$ can be written as the area below the graph of $P[t\leq T],$
\begin{align}
\label{eq:first_moment_identity}
    \mu=t_1 = \int_0^\infty dt\, t p_\text{tick}(t) = \int_0^\infty dt\, P[t\leq T].
\end{align}
The second moment $t_2,$ on the other hand, is related to the center of mass of the graph for $P[t\leq T]$, up to normalization,
\begin{align}
\label{eq:second_moment_identity}
    t_2 = \int_0^\infty dt\, t^2 p_\text{tick}(t) = 2\int_0^\infty dt\, t P[t\leq T].
\end{align}

\subsection{Proof construction}
\label{sec:upper_bound_proof}
The introduction of the Heaviside population in Definition \ref{def:heavisidefamily} together with the result in Lemma~\ref{lemma:heavisideaccuracyresolution}, that these populations achieved the (claimed) optimal accuracy resolution relation, leads us to the following approach in proving the inequality $N\leq\Gamma^2/\nu^2$: we try to show that all clocks are worse than the one with a Heaviside population, which has essentially a perfect background clock. In that sense, what we are showing is that no clock is better than the one that is already perfect. The only premise we have for the proof is that the clocks we consider eventually tick (i.e., $p_\text{tick}(t)$ is a valid probability density) and that their evolution is continuous,\footnote{This ensures we are doing a fair comparison, otherwise, we'd have to discuss how to compare clocks that possibly never tick to ones that always tick. We reserve that discussion for future work.} we call this \textit{well-behaved}. That being said, our strategy to prove the upper bound consists of the following three steps:
\begin{enumerate}[label=(\roman*)]
    \item We show that for any well-behaved top-level population $p(t)$, there exists a Heaviside-$\Theta$ type top-level population with the same average tick time.
    \item Then, we argue that the variance of the tick probability density coming from the Heaviside-$\Theta$ top-level population lower-bounds the variance coming from the generic population $p(t)$.
    \item We conclude that any well-behaved top-level population $p(t)$ must have an accuracy upper bounded by that of the Heaviside population and by using Lemma~\ref{lemma:heavisideaccuracyresolution} on the properties of the Heaviside population, we have $N\leq \Gamma^2/\nu^2$. Once we are there, we have proven the accuracy-resolution trade-off theorem.
\end{enumerate}
\paragraph*{Step (i).}
Begin with a generic, but well behaved top-level population $p(t).$ Let $\mu$ be its average tick time and define $t_0:=\mu-1/\Gamma$. The top-level population
\begin{align}
\label{eq:definition_heaviside_toplevel}
    p^\Theta(t)=\Theta(t-t_0)
\end{align}
has the same first moment $t_1=\mu$ as the generic top-level population.\footnote{The average tick time is given by $\mu^\Theta=t_0+1/\Gamma,$ with $t_0$ coming from the time translation of the Heaviside-$\Theta$ function and $1/\Gamma$ coming from the exponential decay.} For the results from Lemma~\ref{lemma:heavisideaccuracyresolution} to carry over, we need $t_0>0.$ This is generally true as Lemma~\ref{lemma:resolutionbound} guarantees. In fact, it tells us (see Fig.~\ref{fig:cumulative_sketch}) that the average tick time of any well-behaved top-level probability can not be smaller than $1/\Gamma$ and that the best resolution is achieved by the top-level population which is constantly one.
\begin{restatable}[Resolution upper bound]{lemma}{resolutionbound}
\label{lemma:resolutionbound}
For any well-behaved top-level population $p : \mathbb{R}_{\geq 0}\rightarrow [0,1]$, the induced resolution $\nu$ cannot be greater than $\Gamma$.
\end{restatable}
\begin{proof}
Equivalently to the statement in the Lemma, we can prove that $\mu\geq 1/\Gamma.$ For this matter, use eq.~\eqref{eq:first_moment_identity}, to estimate the average tick time
\begin{align}
    \mu = \int_0^\infty dt\, P[t\leq T]\geq \int_0^\infty dt\, e^{-\Gamma t}=1/\Gamma.
\end{align}
This concludes the proof.
\end{proof}
\begin{figure}
    \centering
    \includegraphics[width=\columnwidth]{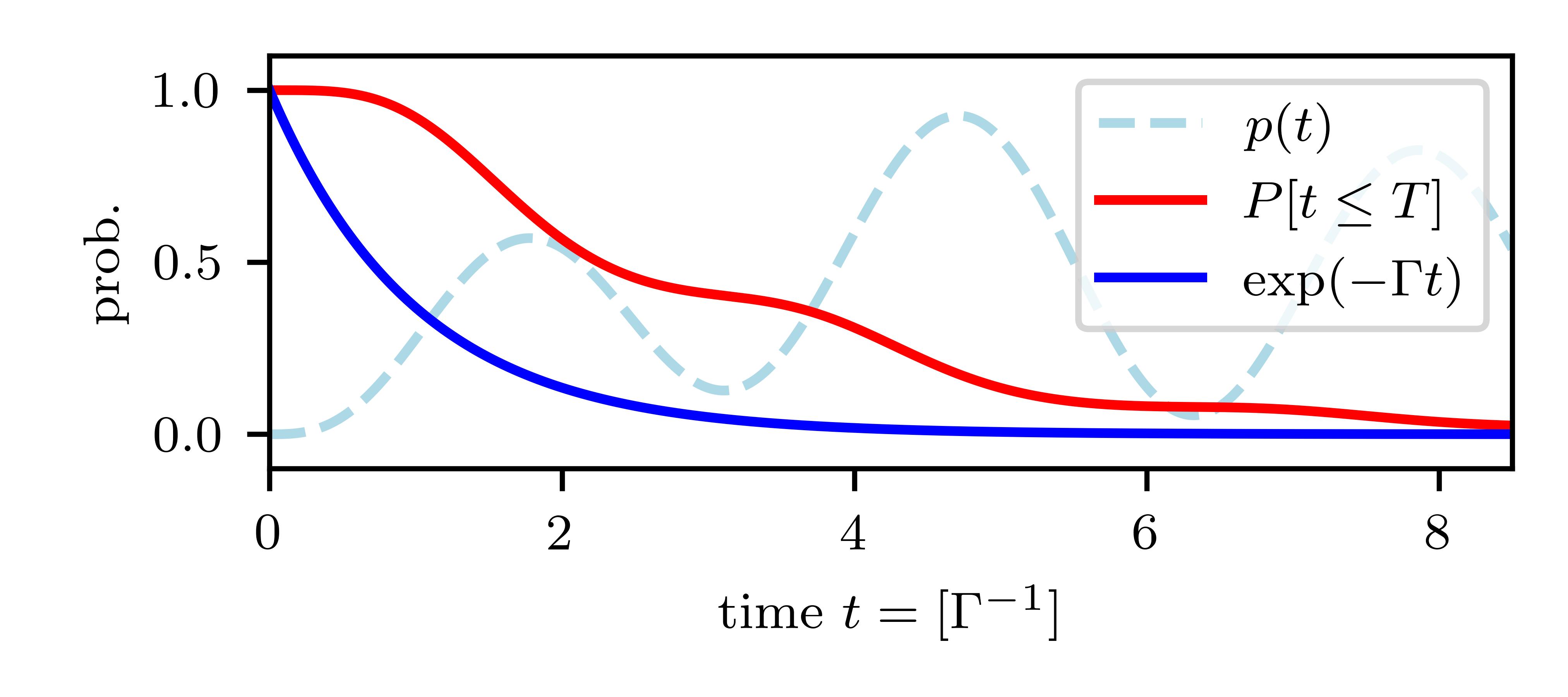}
    \caption{The sketch shows the time-dependency of a generic top-level population $p(t)$ together with the associated cumulative non-tick probability $P[t\leq T]$ compared to the non-tick probability $e^{-\Gamma t}$ of the top-level population which is constantly $1$. A top-level population smaller than unity leads to a slower decay of the non-tick probability, which is responsible for a lower bound on the average tick-time given by $1/\Gamma$.}
    \label{fig:cumulative_sketch}
\end{figure}
\paragraph*{Step (ii).}
If $p(t)= p^\Theta (t),$ there is nothing to be shown, because $p^\Theta(t)$ achieves the minimal tick time variance $\sigma^2=1/\Gamma^2$ (see eq.~\eqref{eq:variance_heaviside}). Hence, we assume from now on $p(t)\neq p^\Theta (t)$ to avoid this pathological case. Proposition~\ref{prop:coincidenceprop} in Section \ref{sec:heaviside_details} further discusses, that there exists a unique $t_*>t_0=\mu-1/\Gamma,$ such that $P[t_*\leq T]=e^{-\Gamma(t_*-t_0)}$ (see Fig.~\ref{fig:cumulative_area}). For all times $t<t_*,$ the non-tick probability for the generic top-level population is smaller than that of the Heaviside population, i.e., $P[t\leq T]\leq P_\Theta[t\leq T].$ In the generic case, a tick before $t_*$ is therefore more likely than for the Heaviside case. On the other hand, for all $t>t_*,$ we have $P[t\leq T]\geq P_\Theta[t\leq T].$ That is, in the generic case, it is also more likely that after $t_*$ the tick did not happen. Formally, the variance of the tick signal is bigger for a generic population than for the Heaviside one (Proposition~\ref{prop:genericvariance}).
\begin{restatable}{proposition}{genericvariance}
\label{prop:genericvariance}
For any well-behaved top-level population $p : \mathbb{R}_{\geq 0}\rightarrow [0,1]$, the variance of the tick probability density $p_\text{tick}(t)$ is lower-bounded by that of $p_\text{tick}^\Theta(t),$ the tick probability density coming from the Heaviside top-level population $p^\Theta(t)$ defined in eq.~\eqref{eq:definition_heaviside_toplevel}.
\end{restatable}
\begin{proof} By construction both $p_\text{tick}(t)$ and $p_\text{tick}^\Theta(t)$ have the same average tick time $\mu$. Because the variance of the tick time is given by the difference $\sigma^2=t_2-\mu^2$, the following two statements are equivalent,
\begin{align}
    \sigma^2 \geq \sigma^2_\Theta \Leftrightarrow t_2 \geq t_2^\Theta.
\end{align}
The second moments on the right-hand side can be calculated as the center of mass of the respective non-tick probabilities $P[t\leq T]$ and $P_\Theta[t\leq T]$ (up to constant prefactors, see also eq.~\eqref{eq:second_moment_identity}). Therefore, the question we have to answer is, whether the graph of $P[t\leq T]$ has a center of mass at larger values of $t$ then that of $P_\Theta[t\leq T]?$
The two areas $A_1$ and $A_2$ in Fig.~\ref{fig:cumulative_area}
\begin{figure}
    \centering
    \includegraphics[width=\columnwidth]{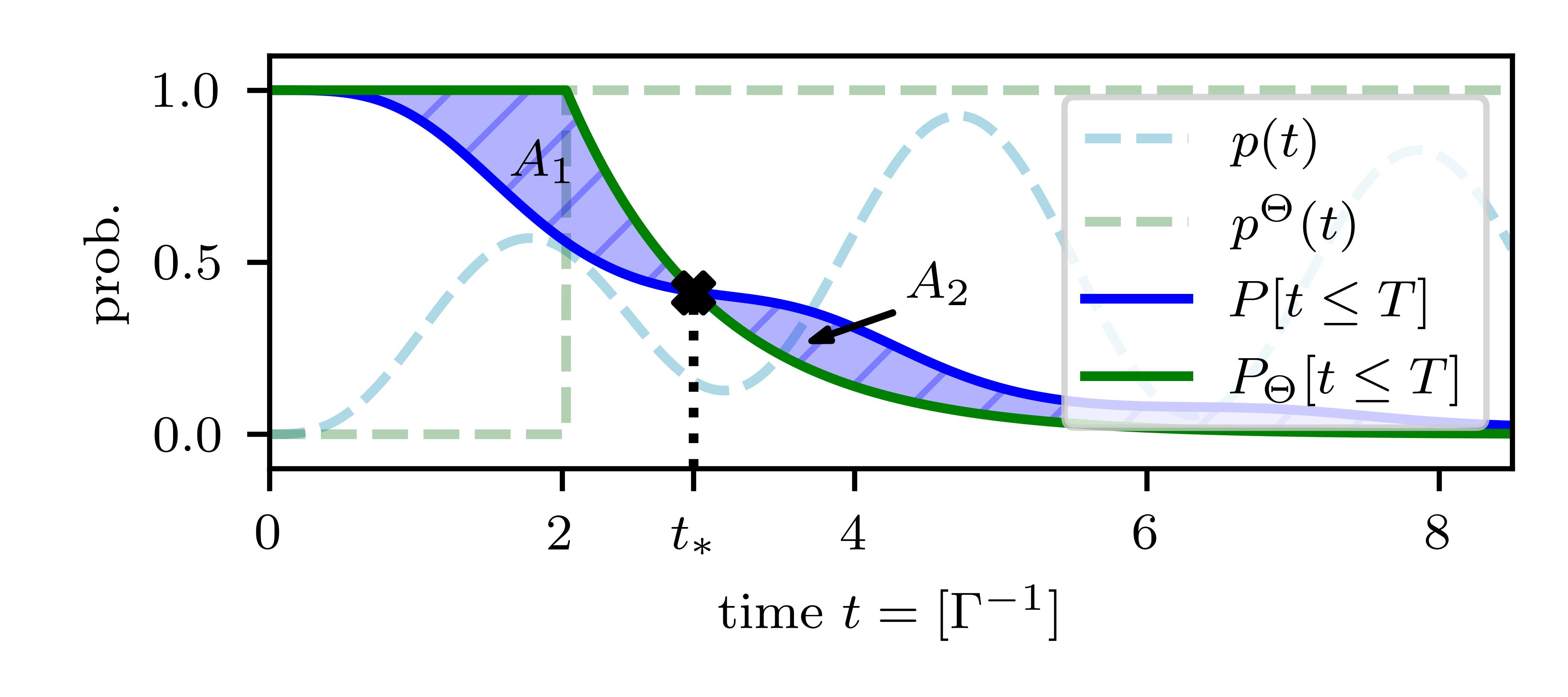}
    \caption{The dashed lines show the top-level populations for the generic (blue) and the Heaviside case (green). The induced cumulative non-tick probabilities $P[t\leq T]$ and $P_\Theta[t\leq T]$ are drawn in the same color but as solid lines. The time $t_*$ marks the unique crossing point of the two non-tick probabilities. On its right and left, the two areas $A_1$ and $A_2$ are defined by the difference between the two non-tick probabilities.}
    \label{fig:cumulative_area}
\end{figure}
are equal and they correspond to the difference in the area below the graphs of $P[t\leq T]$ and $P_\Theta[t\leq T]$.
The graph of $P[t\leq T]$ has less `mass' (i.e., area) on the interval $0\geq t \geq t_*$ than $P_\Theta[t\leq T]$.
For $t_*\geq t \geq \infty$, however, $P[t\leq T]$ has more `mass', resulting in a center of mass further to the right for $P[t\leq T]$ than for $P_\Theta[t\leq T].$
Formally,
\begin{align}
    \frac{t_2}{2}&=\int_0^\infty dt\, t P[t\leq T] \\
    &= \int_0^\infty dt\, t P_\Theta[t\leq T] + \int_0^\infty dt\, t \underbrace{\left(P[t\leq T]-P_\Theta[t\leq T]\right)}_{:=\Delta P(t)}\\
    &\geq \int_0^\infty dt\, t P_\Theta[t\leq T] = \frac{t_2^\Theta}{2},
\end{align}
where the inequality comes about because the integral over $\Delta P(t)$ is greater equal than zero,
\begin{align}
    \int_0^\infty dt\, t \Delta P(t)
    &=\int_0^{t_*} dt\, t \Delta P(t)+\int_{t_*}^\infty dt\, t \Delta P(t)\\
    &\geq \int_0^{t_*} dt\, t_* \Delta P(t) + \int_{t_*}^\infty dt\, t_* \Delta P(t)=0.
\end{align}
All in all, this shows that $t_2\geq t_2^\Theta,$ which, by the initial remark, proves the proposition.
\end{proof}
\paragraph*{Step (iii).}
Combining all the results from the previous two steps, we proof the Theorem.
\begin{proof}
The accuracy $N$ is defined as $N=(\mu/\sigma)^2,$ and the resolution $\nu$ as $\nu=1/\mu.$ We can therefore express
\begin{align}
    N=\left(\frac{\mu}{\sigma}\right)^2=\frac{1}{\nu^2\sigma^2}.
\end{align}
From Proposition~\ref{prop:genericvariance} we know that $\sigma^2\geq \sigma^2_\Theta$, but $\sigma^2_\Theta = 1/\Gamma^2,$ as we derived in Lemma~\ref{lemma:heavisideaccuracyresolution} (see eq.~\eqref{eq:variance_heaviside} in the proof). In conclusion,
\begin{align}
    N=\frac{1}{\nu^2 \sigma^2}\leq \frac{1}{\nu^2 \sigma^2_\Theta}=\frac{\Gamma^2}{\nu^2},
\end{align}
which proves the claim, and thereby the optimality of the family of Heaviside-$\Theta$ top-level populations.
\end{proof}

\subsection{Details on the optimality of the Heaviside population}
\label{sec:heaviside_details}
In Step (ii) of Section \ref{sec:upper_bound_proof}, the claim is made that there exists a unique $t_*\geq t_0=\mu-1/\Gamma,$ such that the two non-tick probabilities coincide $P[t_*\leq T]=e^{-\Gamma(t_*-t_0)}.$ To be more precise, we formulate
\begin{restatable}{proposition}{coincidenceprop}
\label{prop:coincidenceprop}
The set $\mathcal{S}=\{t\geq 0 : P[t\leq T]=  P_\Theta[t\leq T]\}$ is of the form
\begin{equation}
    \mathcal{S}=[0,a] \sqcup \{t_*\},
\end{equation}
where $t_*>t_0,$ and $[0,a]$ is an interval with upper bound $a<t_0.$
\end{restatable}
\begin{figure}
    \centering
    \includegraphics[width=\columnwidth]{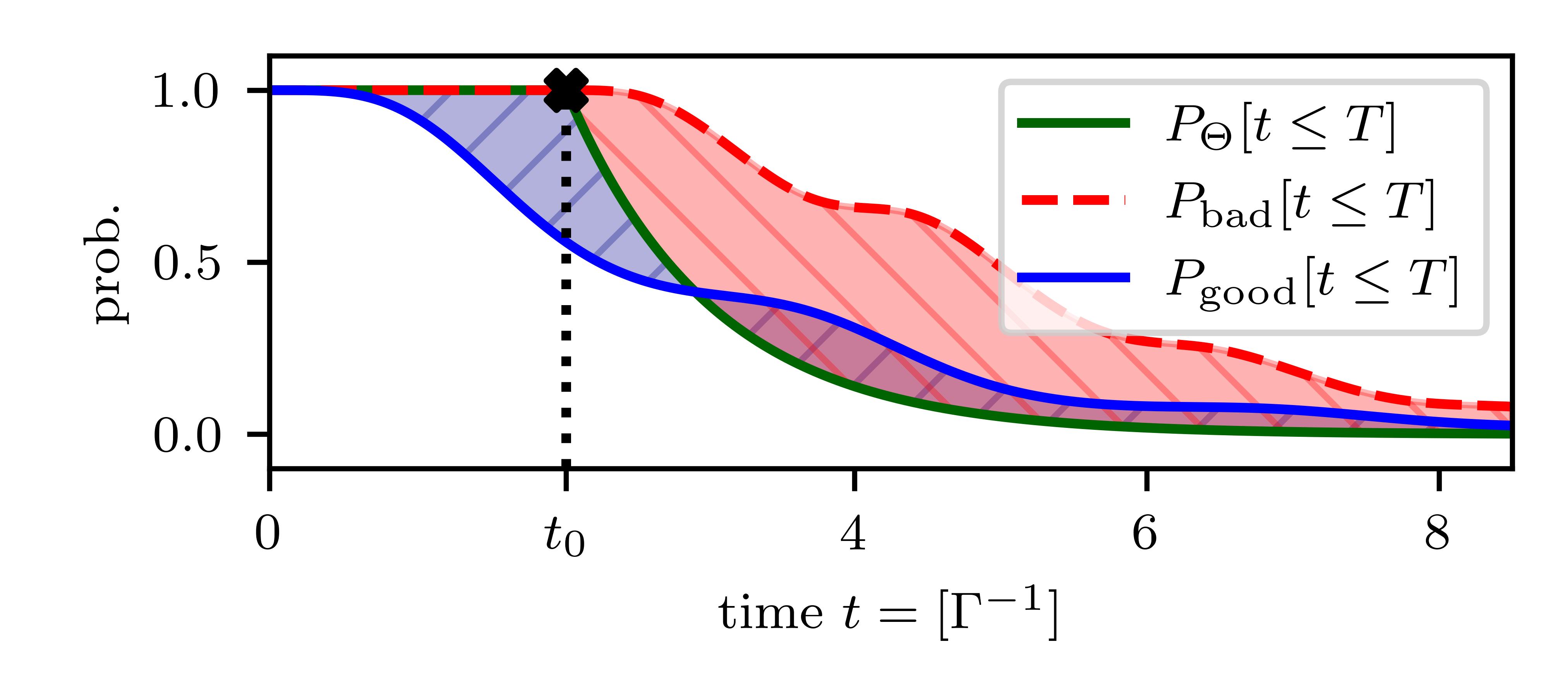}
    \caption{This sketch visualizes the first step of the the proof of Proposition~\ref{prop:coincidenceprop}. The solid green line indicates the non-tick probability for the Heaviside population. We see that for the red dashed line which does not fall below $1$ before $t_0,$ the total integral is greater than that of $P_\Theta[t\leq T],$ contradicting the construction, which ensures that both functions have the same integral. The blue, dashed line shows a plausible function for $P[t\leq T].$}
    \label{fig:a1_step1}
\end{figure}
\begin{figure}
    \centering
    \includegraphics[width=\columnwidth]{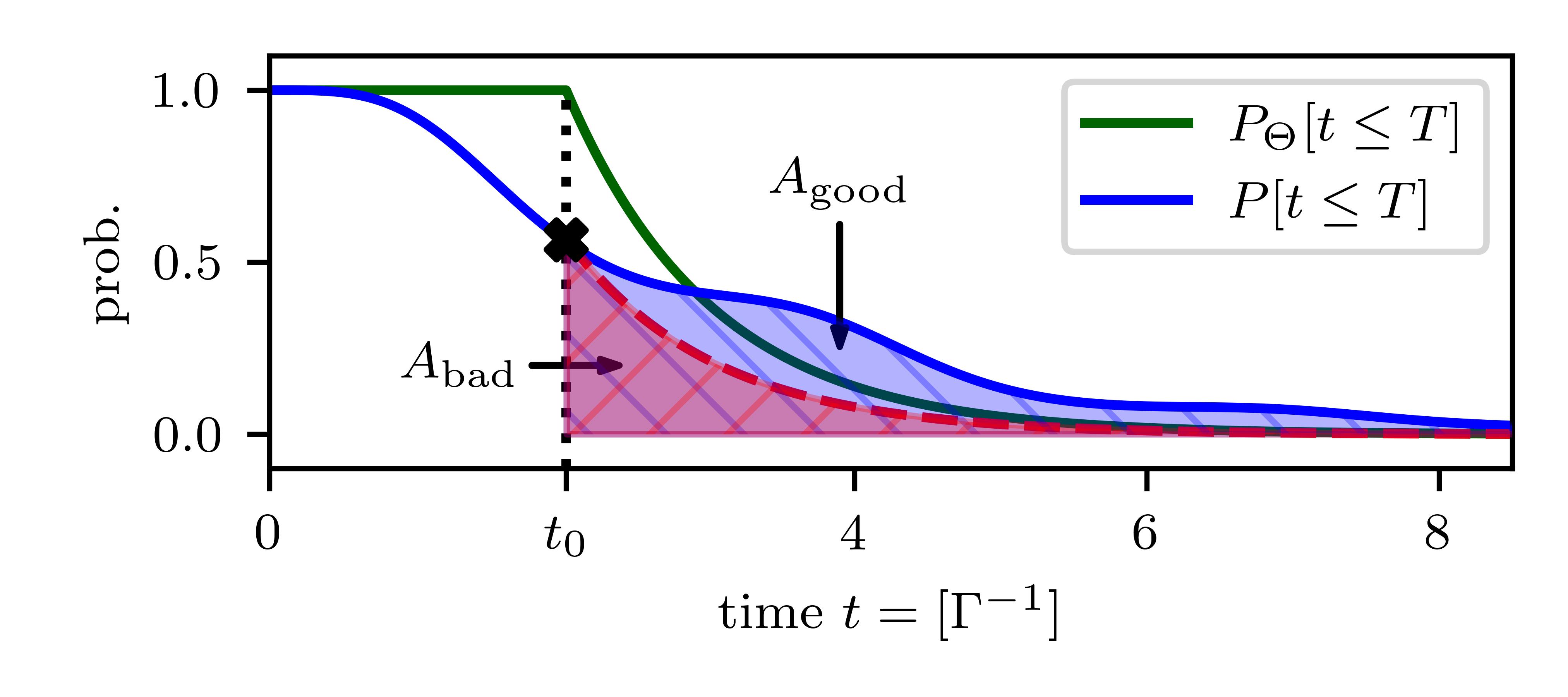}
    \caption{Here, we visualize step two of the proof of Proposition~\ref{prop:coincidenceprop}. Again, the solid, green line indicates the non-tick probability for the Heaviside population. To satisfy the construction that $P[t\leq T]$ and $P_\Theta[t\leq T]$ have the same integral on $\mathbb{R}_{\geq 0},$ the area below the graph of $P[t\leq T]$ for $t>t_0$ must be bigger than that of $P_\Theta[t\leq T].$ Thus, at some point after $t_0,$ $P[t\leq T]$ must be bigger than $P_\Theta[t\leq T]$ (see the blue area denoted by $A_\text{good}$), and not always smaller (see red area denoted by $A_\text{bad}$).}
    \label{fig:a2_step3}
\end{figure}
\begin{proof} We divide the proof of this proposition into three steps (a) - (c), with visual aids in Figures \ref{fig:a1_step1} and \ref{fig:a2_step3}.
\paragraph*{Step (a).} We claim that \textit{there exists a $0< t< t_0,$ such that $P[t\leq T]<1.$}
Suppose that this was not the case, then $P[t\leq T]=1$ for all $t\leq t_0$ (see Fig.~\ref{fig:a1_step1}).
By requirement, there must exists a $t$ such that $p(t)\neq p^\Theta(t),$ and because by our (contradictory) assumption $t$ cannot be smaller than $t_0,$ it must be bigger than $t_0$. Since, $P[t\leq T]\geq e^{-\Gamma(t-t_0)},$ for all $t>t_0,$ there must exist a $t>t_0,$ such that $P[t\leq T]> e^{-\Gamma(t-t_0)}$. But this contradicts the assumption that $\int_0^\infty dt P[t\leq T]=t_0+1/\Gamma.$ Therefore, it must be that there is a $t<t_0$ with $P[t\leq T]<1.$
\paragraph*{Step (b).} Our next claim is, that \textit{there exists an $0<a<t_0$ such that on $[0,a],$ we have $P[t\leq T]=1$ and for all other $a<t\leq t_0,$ $P[t\leq T]<1$.} By our previous claim (1), there exists $0<t<t_0$ with $P[t\leq T]<1.$ Let $a$ be the infimum of all such $t$.\footnote{The set $\{t : P[t\leq T] < 1\}$ is lower bounded by $0$ and non-empty by the claim made in step (a), thus, the infimum exists.} Then, $[0,a]\subset \mathcal{S}$ (continuity of $P$). Moreover, by the properties of an infimum, for all $\varepsilon>0,$ there exists a $\delta>0,$ such that $P[a+\delta\leq T]<1$. But $P[t\leq T]$ is monotonically decreasing, hence, this is true for all $\delta>0,$ such that $a+\delta\leq t_0$, proving the claim.
\paragraph*{Step (c).} In this last step, we show \textit{the existence of a unique $t_*>t_0$ such that $t_*\in \mathcal{S}.$} For a visualization of this step, see Fig.~\ref{fig:a2_step3}.
The integrals of $P[t\leq T]$ and $P_\Theta[t\leq T]$ on $\mathbb{R}_{\geq 0}$ coincide. However, we have just argued that on a non-empty interval between $0$ and $t_0,$ $P[t\leq T]$ is strictly smaller than $P_\Theta[t\leq T].$ This leads to the the inequality $\int_0^{t_0}dt\,P[t\leq T]<\int_0^{t_0}dt\,P_\Theta[t\leq T].$ In order to ensure that the integrals over all of $\mathbb{R}_{\geq 0}$ are equal, there must exist a $t'>t_0$ such that $P[t'\leq T]>P_\Theta[t'\leq T].$ For all $t>t',$ this inequality must be true too, because $P[t\leq T]$ can not drop faster than $e^{-\Gamma(t-t_0)},$ due to boundedness of the top-level population $p(t)\leq 1.$ Set $t_*$ to be the infimum of all such $t'.$ Note that $t_*$ must be strictly greater than $t_0$ due to continuity of $P[t\leq T].$ This concludes the proof of the last step and therefore of the proposition.
\end{proof}

\end{document}